\documentclass[conference]{IEEEtran}

\usepackage[utf8]{inputenc}
\usepackage{amsmath,amssymb,bm}
\usepackage[dvips]{graphicx}
\usepackage{array}
\usepackage{xparse}
\usepackage[section]{placeins}
\usepackage{cite}
\usepackage{empheq}
\usepackage{algpseudocode}
\usepackage{enumerate}
\usepackage{multicol}
\usepackage{lipsum}
\usepackage{color}
\usepackage{float}
\usepackage[capitalise]{cleveref}
\usepackage[mathcal]{euscript}
\usepackage{multirow,multicol,tabu}
\usepackage{subcaption}
\usepackage{soul}

\newcolumntype{?}{!{\vrule width 1pt}}

\usepackage{mathtools}	
 
\usepackage{algorithm}
\usepackage{algpseudocode}

\usepackage{tikz,tabularx}
\usetikzlibrary{decorations.pathreplacing}
\tikzstyle{every picture}+=[remember picture]

\usetikzlibrary{shapes,arrows,shadows}	
\usepgflibrary{patterns}
\usetikzlibrary{patterns}
\usetikzlibrary{shapes.geometric}
\usetikzlibrary{arrows}

\usetikzlibrary{decorations.markings}

\tikzset{myptr1/.style={decoration={markings,mark=at position 1 with %
    {\arrow[scale=2.5]{>}}},postaction={decorate}}}
    
\tikzset{myptr2/.style={decoration={markings,mark=at position 1 with %
    {\arrow[scale=1.8]{>}}},postaction={decorate}}}
    
\usetikzlibrary{decorations.markings}

\usepackage{multicol}
\usepackage{amsthm}

\newtheorem{remark}{Remark}

\newtheorem{theorem}{Theorem}
\newtheorem{lemma}{Lemma}
\newtheorem{proposition}{Proposition}

\def \DesignedTimeA {$55$}
\def \DesignedTimeB {$30$}
\def \DesignedTimeC {$55$}
\def \DesignedTimeD {$8$}
\def \DesignedTimeE {$10$}

\def \iid {iid}

\title{Approximate Gradient Coding for Heterogeneous Nodes}

\author{
  	\IEEEauthorblockN{Amogh Johri}
 	\IEEEauthorblockA{IIIT Bangalore}
 	\IEEEauthorblockA{Amogh.Johri@iiitb.org}		
 	\and
 	\IEEEauthorblockN{Arti Yardi}
 	\IEEEauthorblockA{IIIT Bangalore}
 	\IEEEauthorblockA{arti.yardi@iiitb.ac.in}		
 	\and
 	\IEEEauthorblockN{Tejas Bodas}
 	\IEEEauthorblockA{IIT Dharwad}
 	\IEEEauthorblockA{tejaspbodas@iitdh.ac.in}		
}

\begin{document}

\maketitle

\begin{abstract}
In distributed machine learning (DML), the training data is distributed across multiple worker nodes to perform the underlying training in parallel. One major problem affecting the performance of DML algorithms is presence of stragglers. These are nodes that are terribly slow in performing their task which results in under-utilization of the training data that is stored in them. Towards this, gradient coding mitigates the impact of stragglers by adding sufficient redundancy in the data. Gradient coding and other straggler mitigation schemes assume that the straggler behavior of the worker nodes is identical. Our experiments on the Amazon AWS cluster however suggest otherwise and we see that there is a correlation in the straggler behavior across iterations. To model this, we introduce a heterogeneous straggler model where nodes are categorized into two classes, slow and active. To better utilize training data stored with slow nodes, we modify the existing gradient coding schemes with shuffling of the training data among workers. Our results (both simulation and cloud experiments) suggest remarkable improvement with shuffling over existing schemes. We perform theoretical analysis for the proposed models justifying their utility.      
\end{abstract}
%


%
\section{Introduction}
\label{Section_Introduction}

%


In distributed computing, a computation task is distributed among a set of $n$ machines or worker nodes that perform a partial computation of the original task in parallel. The parallel computations are then combined by the master node to produce the desired result in an efficient way.
Distributed computing has the ability to increase the speed of 
parallelizable algorithms by a significant factor, by processing
each parallel division on a different worker node.
However, a major bottleneck in DML is that occasionally, some worker nodes are slower than the others, termed as \textit{stragglers}, and these stragglers affect the speed of overall computation to a great extent. 
Towards mitigating the impact of stragglers in distributed computing, coding theoretic techniques, have been found to be particularly useful~\cite{Kannan_first_DML_2018, Exact_GD_First_2017}. 

%
In this paper, we focus on distributed machine learning setup,
where the aim is to implement the iterative gradient descent algorithm. 
Coding techniques used in this setup are termed as \textit{gradient coding}~\cite{Exact_GD_First_2017, RS_Grad_2018, Gunduz_partial_2020, Gunduz_multimessage_2019,Gunduz_non_persistent_2018, 
Approx_GD_Expander_graph_Raviv_2017, Apprx_GD_sparse_graph_Papailiopoulos_2017_arxiv, Apprx_GD_Erasurehead_Papailiopoulos_2019_arxiv}. 
%
%
%
In gradient coding, the key idea is to create data partitions with coded redundancy such that they are robust to stragglers. 
In the literature, two type of models have been proposed to account for the straggler behavior, namely the \textit{worst-case} and \textit{average-case} straggler models. 
In the \textit{worst-case} straggler model, the maximum number of stragglers $s_{max}$ is assumed to be known and a coding scheme is devised such that the 
master is able to compute the exact gradient when the number of actual stragglers is less than or equal 
to $s_{max}$, termed as $s_{max}$-tolerant scheme~\cite{Exact_GD_First_2017, RS_Grad_2018}.
The \textit{average-case} model corresponds to the case where the number of actual stragglers can be more than $s_{max}$ and the master might not always be able to compute the exact value of the gradient. In this case, the master computes an approximate gradient using the computations performed by the set of 
non-stragglers~\cite{Apprx_GD_sparse_graph_Papailiopoulos_2017_arxiv, Approx_GD_funda_limits_Shroff_2019, Apprx_GD_Erasurehead_Papailiopoulos_2019_arxiv, Approx_GD_SGD_Bitar_2020, Apprx_GD_optimal_Decoding_Glasgow_2020_arxiv}.
In many applications, the average-case straggler model is desirable as $s_{max}$ might not be known \textit{a priori}.

\input{Figures/Figure_Straggler_behaviour_AWS_ISIT}

%
Irrespective of the model, a typical assumption in the literature is that each worker could behave as a straggler with a probability which is independent and identically distributed (\iid) across workers and across iterations of the algorithm. 
Since the effectiveness of the proposed method 
depends heavily on the validity of this assumption, as a first step, we performed extensive experiments on Amazon EC2 clusters using $t2.micro$ nodes to study the straggler behavior.
%
%
We ran a dummy program that should ideally take a deterministic time at each worker.
We observed that a worker node either behaves as a straggler or a non-straggler throughout the course of the experiment.
We see that while the straggler behavior is indeed \iid~across workers, it is heavily correlated across iterations of a worker (see \cref{Figure_AWS_straggler_example}). 
Motivated by our experiments, we propose a new straggler model, termed as \textit{heterogeneous straggler model} that aims to mimic stragglers in a better way. 
In this model, we classify workers into two classes, slow and active-class such that the slow-class workers are more likely to straggle. Further, each worker retains its class throughout the experiment.
See \cite{Timely_coded_Avestimehr_ISIT_2019, Timely_coded_Avestimehr_2019} for a recent attempt at modeling straggler nodes using a Markov chain. In this model, each worker can either be in a \textit{good} state or \textit{bad} state and the transitions between the two states are characterized using a Markov chain.
The schemes proposed in these works consist of estimating
the state transition probabilities with the aim of optimizing the data-sets distribution among the workers.
This model however does not accurately justify our observations that nodes are inherently either slow or active throughout the course of the experiment.

For our heterogeneous straggler model, it is intuitive to see that a data-set which is assigned to a slow-class worker is accessed less number of times compared to the data-set assigned to an active-class worker 
(see \cref{Proposition_FRC_CRC_with_without_shuffling}).
Due to this, we observe via simulations that, the coding schemes designed for iid stragglers perform poorly on heterogeneous model (see \cref{Figure_Loss_curves_ISIT}). A simple remedy for accessing data uniformly would be to shuffle the data randomly among the workers in each iteration of the algorithm. 
Note that this shuffling can be implemented without significant overhead in cloud computing platforms such as Amazon Web Services (AWS)~\cite{AWS_doc} and Microsoft Azure~\cite{WinNT}, where the training data is stored in a common memory that is accessible to all workers.
The main contributions of our work are as follows.
\begin{itemize}
\item We propose a new straggler model based on our experiments on the Amazon EC2 cluster to mimic the correlated straggler behavior across iterations.
\item To make the approximate gradient coding robust under this model, we introduced data shuffling that can be applied on top of
existing schemes such as fractional repetition codes (FRC) and cyclic repetition codes (CRC) (proposed in \cite{Exact_GD_First_2017}).
\item We provide a thorough analysis for these two schemes under shuffling
and obtain expressions for their expected optimal decoding error.
Our analysis for the FRC scheme extends the existing results~\cite[Thm.~6]{Apprx_GD_sparse_graph_Papailiopoulos_2017_arxiv} for the heterogeneous straggler model. We are also the first to analyze the expected optimal decoding error under the CRC scheme, invoking the properties of \textit{cycle class} and \textit{cycle representatives} \cite{Cyclic_decomposition_Tavares_1971}. 

%
\item Finally, we provide numerical results from our experiments on a simulated heterogeneous straggler environment and also from the Amazon EC2 clusters that indicate that shuffling does result in an efficient use of training data under this model.
\end{itemize}

%
\textit{Organization:} In \cref{Section_System_model}, we provide problem setup for our work.
Our proposed shuffling method and its analysis are discussed in Sections~\ref{Section_shuffling_scheme} and \ref{Section_one_iteration_error}.
Numerical results are given in \cref{Section_simulations} with concluding remarks in \cref{Section_Conclusion}.

%
\section{System model}
\label{Section_System_model}


We consider the distributed master-node computing setup consisting of $n$ worker nodes 
$W_1, W_2, \ldots, W_n$ and a master node.
%
Suppose we have been given a set of $N$ data points
$\mathcal{D} = \{(x_1,y_1), (x_2,y_2), \ldots, (x_N,y_N)) \}$ such that each $(x_i,y_i) \in \mathbb{R}^a \times \mathbb{R}$ and $a \in \mathbb{N}$.
This $\mathcal{D}$ can be divided into $n$ disjoint data-sets of equal size, denoted by $D_1, D_2, \ldots, D_n$.
For simplicity, we assume that $n$ divides $N$.
The master wishes to solve the following optimization problem
%
\begin{align}
\beta^* = \arg \min_{\beta \in \mathbb{R}^b} \sum_{i=1}^{N}l(\beta;x_i,y_i) + \lambda R(\beta),
\end{align}
where $b \in \mathbb{N},$ $l(.)$ represents the loss function, $R(.)$ represents the regularization function, $\lambda$ is a constant, and $\beta$ is a parameter that we wish to estimate~\cite{ML_book_Shwartz}.
This problem can be solved using iterative gradient based algorithms, where one begins with the initial guess $\beta^{(0)}$ for the parameter $\beta$ and the $\ell$-th iteration of the algorithm computes the gradient $g$ given by, $g = \sum_{i=1}^N \nabla l(\beta^{(\ell)};x_i,y_i)$ 
%
After every iteration, $\beta^{(\ell+1)}$ is obtained using $\beta^{(\ell+1)} = h_R(\beta^{(\ell)},g)$, where $h_R(.)$ corresponds to the {gradient based optimizer}. 
Suppose we perform a total of $L$ such iterations to estimate $\beta$.
In a distributed compute setting, each worker node $W_j$ computes the partial gradient $g_j$ corresponding to data-set $D_j$ 
given by $g_j = \sum_{(x_j,y_j) \in D_j} \nabla l(\beta^{(\ell)};x_j,y_j)$,
where $j = 1,2,\ldots, n$ and the complete gradient can be obtained by $g = \sum_{j=1}^n g_j$. 
%
We consider the architecture where the data-sets are stored in a common memory that is accessible to all the worker nodes\footnote{Our architecture is different from related literature~\cite{Exact_GD_First_2017,  Apprx_GD_sparse_graph_Papailiopoulos_2017_arxiv, Approx_GD_SGD_Bitar_2020}, where each worker node has access to a strict subset of data-sets.}. 
Note that this architecture is used in many real life distributed computing systems such as Amazon Web Services (AWS)~\cite{AWS_doc},  Microsoft Azure~\cite{WinNT}, and in general Network File Systems~\cite{Shared_memory_Radenski}.
%

%
%
%
%

\subsubsection*{Heterogeneous straggler model}
\label{Subsection_heterogenous_straggler_model}
Workers are divided into two classes, slow and active-class.
A worker node belongs to the slow-class with probability $\hat{p}$  independently of the remaining nodes. 
A worker remains in its assigned class throughout the experiment of $L$ iterations.
We additionally assume that in any iteration, a node in the slow-class straggles with probability $p_{ss}$ and this behavior of the node is independent of any of its previous iteration. Similarly, a node in the active-class straggles with probability $p_{as}$, such that $p_{ss} > p_{as}$. 
Note that in practice, this corresponds to the setup when in each iteration, the master waits for a fixed amount of time for workers to send their computations, and $p_{ss}$ and $p_{as}$ represents the probability with which the slow and active nodes are unable to perform their computation in the stipulated time. 

%
%
%


\subsubsection*{Some preliminaries about gradient coding}
\label{Subsection_gradient_coding_preliminaries}

In gradient coding, each worker computes partial gradients corresponding to more than one data-set~\cite{Exact_GD_First_2017}.
An encoding matrix $B \in R^{n \times n}$ represents which worker node gets assigned to which data-sets. 
The rows of $B$ correspond to $n$ data-sets and the columns correspond to $n$ workers such that $B(i,j) \neq 0$ if $W_j$ computes the partial gradient $g_i$ corresponding to $D_i$.
In each iteration, $W_j$ sends the linear combination 
$\sum_{i=1}^n B(i,j) g_{i}$ of the partial gradients that it has computed to the master.
Suppose coding scheme $B$ is designed to tolerate $(s-1)$  stragglers. 
%
%
Suppose in an iteration, the master node received gradients from $r$ non-stragglers where $r\in \{0, 1, \ldots, n\}$.
Let $A \in \mathbb{R}^{n \times r}$ be the submatrix is formed by considering the columns of $B$ that correspond to non-stragglers. 
This matrix $A$ is termed as a \textit{non-straggler matrix}~\cite{Apprx_GD_sparse_graph_Papailiopoulos_2017_arxiv}.
For an $(s-1)$-tolerant coding scheme, master is guaranteed to compute the exact gradient $g$ when $r \geq n - s + 1$~\cite{Exact_GD_First_2017}.
When $r < n-s+1$, master may or may not be able to compute $g$ exactly and an error is said to have occurred. 
This error is termed as the \textit{optimal decoding error} $\mbox{err}(A)$ associated to $A$ and it is given by~\cite{Apprx_GD_sparse_graph_Papailiopoulos_2017_arxiv}
\begin{align}
%
\mbox{err}(A) \coloneqq \min_{\mathbf{x} \in \mathbb{R}^{r}} || A \mathbf{x} - \mathbf{1}_n ||_2^2, 
\label{Eqn_Error_submatrix_A}
\end{align}
where $\mathbf{1}_n$ denotes the all-ones vector of length $n$.


\section{Proposed shuffling scheme}
\label{Section_shuffling_scheme}
%
Recall that in our heterogeneous straggler model, workers can either belong to a slow-class or an active-class and  
workers retain their classes throughout the experiment of $L$ iterations.
In any $\ell$-th iteration we say that a data-set $D_i$, $1 \leq i \leq n$, is \textit{accessed} if there exists a non-straggler worker $W_j$ to which $D_i$ is assigned, i.e., $A(i,j) \neq 0$ for some $j$, $1 \leq j \leq r$, where recall that $r$ is the number of non-stragglers and  $A$ is the corresponding non-straggler matrix.
For our heterogeneous straggler model, one can intuitively see that, on an average, a data-set which is assigned to a slow-class worker is accessed less number of times compared to the data-set which is assigned to an active-class worker (this is formally proved in \cref{Proposition_FRC_CRC_with_without_shuffling}). 
When data-sets are accessed non-uniformly, this leads to poor training of the parameter $\beta$, compared to the situation when all data-sets are accessed uniformly (see \cref{Figure_Loss_curves_ISIT} for numerical evidence).

Therefore it is essential to design a new coding scheme to determine the data partitioning between slow and active nodes.
One possible scheme could be designed based on \cite{Timely_coded_Avestimehr_ISIT_2019}, where the parameters $\hat{p}, p_{as}, p_{ss}$ are estimated and then data-sets are distributed based on whether nodes are slow or active. However such estimation techniques offer an overhead to the system and may not lead to desired results if these straggling probabilities change over the time. Further, for the experiments of shorter iterations, these estimates might be far from accurate leading to suboptimal behavior. To alleviate these difficulties, we seek a simple yet effective strategy that would prevent uneven data-set access.

Note that our system architecture is such that all the data-sets are stored in a shared memory that is accessible to all workers. Instead of designing a new coding scheme, we leverage the common memory architecture to our advantage to ensure that the data sets are accessed uniformly. We observe that by simply shuffling the columns of the encoding matrix $B$ randomly in each iteration, one can access the data-sets as uniformly under the heterogeneous straggler model.
We prove this for FRC and CRC schemes~\cite{Exact_GD_First_2017} in the following elementary proposition and discuss our main results on the optimal decoding error in the next section. 
Note that for both $(s-1)$-tolerant FRC and CRC schemes, any data-set is assigned to exactly $s$ workers~\cite{Exact_GD_First_2017}.
%

\begin{proposition}
\label{Proposition_FRC_CRC_with_without_shuffling}
For $(s-1)$-tolerant FRC (or CRC) scheme, suppose  
%
in the first iteration, data-set $D_i$, $1 \leq i \leq n$, is assigned to workers $W_{j_1}, W_{j_2}, \ldots, W_{j_s}$, where $1 \leq j_1, j_2, \ldots, j_s \leq n$.
%
Suppose that exactly $k_i$ workers out of these $s$ belong to the slow-class, where $0 \leq k_i \leq s$.
Then in an experiment of $L$ iterations, the expected number of times $D_i$ is not accessed with and without shuffling is equal to $L \big(\hat{p}p_{ss} + (1-\hat{p})p_{as} \big)^s$ and $L p_{ss}^{k_i}p_{as}^{s-{k_i}}$ respectively.
\end{proposition}
\begin{proof}
Let $A$ be the event that $D_i$ is not accessed in any $\ell$-th iteration.
When the columns of $B$ are not shuffled, workers assigned to $D_i$ do not change over the course of $L$ iterations
and since exactly $k_i$ workers belong to the slow-class, $\mathbb{P}[A|\mbox{no-shuffling}] = p_{ss}^{k_i}p_{as}^{s-k_i}$.
When the columns of $B$ are shuffled uniformly at random in each iteration, essentially the set of workers assigned to $D_i$ are shuffled uniformly.
Thus in any $\ell$-th iteration, any worker that is assigned to $D_i$ will belong to the slow and active-class with probabilities $\hat{p}$ and $1-\hat{p}$ respectively.
For event $A$ to occur, all $s$ workers that are assigned to $D_i$ should straggle and hence $\mathbb{P}[A|\mbox{shuffling}] = (\hat{p}p_{ss} + (1-\hat{p})p_{as} \big)^s$.
In heterogeneous straggler model, any worker straggles independently in each iteration (with probabilities decided according their classes). Hence event $A$ 
%
follows the binomial distribution 
%
and this gives us the required expectations.
%
\end{proof}

It can be seen that, while for the non-shuffling scheme the expected number of times $D_i$ is not accessed is a function of number of slow-class workers ($k_i$) assigned to $D_i$, for the shuffling scheme it is not. This causes non-uniform data access for the non-shuffling scheme.

\begin{remark}
Random shuffling has a disadvantage that in each iteration, the master needs to provide each worker node with the address of the memory locations of the data-sets that a worker is supposed to work on in that iteration. 
However, this disadvantage can be mitigated by using a pre-decided shuffling pattern, such as cyclic shuffling.
%
%
When $L \gg n$, pre-decided shuffling patterns also give desirable results.
\hfill $\square$
\end{remark}

\section{Main results}
\label{Section_one_iteration_error}
%
%
In this section, we find the expected decoding error in any $\ell$-th iteration with our shuffling scheme for FRC and CRC.
We assume that $m$ workers out of $n$ belong to slow-class and refer our model as $(n,m)$-heterogeneous straggler model.
%

\subsection{Expected decoding error for FRC with shuffling scheme}
\label{subsection_FRC_decoding_error}

The expected error for the FRC scheme in presence of iid stragglers
is studied in \cite[Theorem~6]{Apprx_GD_sparse_graph_Papailiopoulos_2017_arxiv}. We now extend this analysis for our $(n,m)$-heterogeneous straggler model.
%

\begin{theorem}
\label{Theorem_FRC_EX_error_one_iteration}
For $(n,m)$-heterogeneous straggler model, the expected error $\mathbb{E}_{err}(FRC, n, s)$ for an $(s-1)$-tolerant, $n$-node FRC scheme with random shuffling  is given by
\begin{align*}
\mathbb{E}_{err}(FRC, n, m, s) 
= n \sum_{r = 0}^{n-s} P_r  {n-s \choose r}\Big/{n \choose r} , 
\end{align*}
where $P_r$ is the probability that there are $r$ non-stragglers in any $\ell$-th iteration and is given by 
\begin{align*}
P_r = \sum_{i=0}^r {m \choose r^{\prime}(i)} {n-m \choose r-r^{\prime}(i)} \bar{p}_{ss}^i p_{ss}^{(m-i)} \bar{p}_{as}^{(r-i)} p_{as}^{(n-m-r+i)}
\end{align*}
where $r^{\prime}(i) = \min(i,m), \bar{p}_{ss} = 1 - p_{ss}$ and $\bar{p}_{as} = 1 - p_{as}$.
%
\end{theorem}
%
%
\begin{proof}
We first find the probability $P_r$ that there are $r$ non-stragglers in any $\ell$-th iteration.
Each non-straggler can either belong to the slow-class or active-class. 
Suppose $i$ non-stragglers out of total $r$ non-stragglers belong to the slow-class, then the remaining $r-i$ should belong to the active-class. 
For our $(n,m)$-heterogeneous straggler model, probability of this is given by $\bar{p}_{ss}^i p_{ss}^{(m-i)} \bar{p}_{as}^{(r-i)} p_{as}^{(n-m-r+i)}$.
Conditioning over $i$ we obtain the expression of $P_r$.

In any $\ell$-th iteration, due to random shuffling, the probability that a column of $B$ is associated with either slow (or active-worker) remains the same for each column and hence all possible ${n \choose r}$ submatrices are equally likely. From \cite[Thm.~6]{Apprx_GD_sparse_graph_Papailiopoulos_2017_arxiv}, 
the expected error under the condition of $r$ non-stragglers is equal to $n {n-s \choose r}/{n \choose r}$(Kindly note that we have corrected Thm.~6 of \cite{Apprx_GD_sparse_graph_Papailiopoulos_2017_arxiv}, where this expression is given to be $n{n-s \choose s-r}/{n \choose r}$). 
The expression of $\mathbb{E}_{err}(FRC, n, m, s)$ now follows by conditioning over $r$ non-stragglers and noting that when $r \geq n-s,$ the decoding error is zero. 
\end{proof}

\subsection{Expected decoding error for CRC with shuffling scheme}
\label{subsection_CRC_decoding_error}

We first introduce some notation that we shall need in this subsection.
Suppose $\mathbb{F}_2^n$ denotes the set of all possible binary vectors of lengths $n$. For a vector $\mathbf{v} \in \mathbb{F}_2^n$, let $\mathcal{S}_{\mathbf{v}}$ denotes the set of distinct vectors that are obtained by considering all possible cyclic shifts of $\mathbf{v}$. 
In the literature, the set $\mathcal{S}_{\mathbf{v}}$ is referred to as the \textit{cycle class} of $\mathbf{v}$ and the weight of class $\mathcal{S}_{\mathbf{v}}$ is defined as the weight of the corresponding $\mathbf{v}$~\cite{Cyclic_decomposition_Tavares_1971}.
Since the elements of a cycle class are cyclic shifts of one another, a cycle class can be represented by any one of its element termed as \textit{cycle representative}. 
The cardinality of the set $\mathcal{S}_{\mathbf{v}}$ is called as its \textit{order}, denoted by $e$ and it is known that $e$ should divide $n$~\cite{Cyclic_decomposition_Tavares_1971}. 
The set $\mathbb{F}_2^n$ can be partitioned into distinct cycle classes and the number of such distinct cycle classes of $\mathbb{F}_2^n$ are characterized in the following lemma. 
\begin{lemma}
\label{Lemma_Tavares}
\cite[pp.11]{Cyclic_decomposition_Tavares_1971}
The number of distinct cyclic classes of $\mathbb{F}_2^n$ of weight $r$ and order $e$, denoted by $N(r,e)$, is equal to the coefficient of $x^r$ in the polynomial $R_e(x)$ given by
\begin{align*}
R_e(x) = \frac{1}{e} \sum_{d|e} \mu\left(\frac{e}{d}\right) (1+x^{n/d})^d
\end{align*}
where $\mu(.)$ is the Moebius function defined on set of positive integers given by~\cite[Ch.~3]{Lidl86}, 
$\mu(n)= 1$ if $n=1$, $\mu(n)= 0$ if $n=1$ if $n$ is divisible by the square of a prime, and $\mu(n)= (-1)^k$ if $n$ is the product of $k$ distinct primes.
%
%
\hfill $\square$
\end{lemma}

For the given $r \in \{0, 1, \ldots, n\}$, we now focus on the vectors of weight $r$. Let $N_r$ be the total number of distinct cycle classes of weight $r$. From \cref{Lemma_Tavares}, we have $N_r = \sum_{e | n} N(r,e)$. 
Let $\mathcal{V}_r \coloneqq \{\mathbf{v}_1,\mathbf{v}_2, \ldots, \mathbf{v}_{N_r} \}$ be the cycle representatives of these $N_r$ cycle classes and $\mathcal{E}_r \coloneqq \{e_1, e_2, \ldots, e_{N_r} \}$ be their respective orders\footnote{Various algorithms for finding the cycle representatives of binary cyclic codes are studied in the literature. For example, one can refer to the algorithms provided in \cite[Page 6]{Cyclic_decomposition_Tavares_1973} and \cite[Section 8.4]{Macwilliams_Sloane_1977}. Since $\mathbb{F}_2^n$ is trivial cyclic code, one can use these algorithms to find its cycle representatives.}.

We now relate these cycle classes to non-straggler matrices of $(s-1)$-tolerant, $n$-node CRC scheme with encoding matrix $B_{CRC}$.
Corresponding to each $\mathbf{v}_i \in \mathcal{V}_r$, let $A_{\mathbf{v}_i}$ be the submatrix of $B_{CRC}$ such that $A_{\mathbf{v}_i}$ consists of the $j$th column of $B_{CRC}$ if the $j$th component of $\mathbf{v}_i$ is equal to one, for $j = 1, 2, \ldots, n$.
For every $\mathbf{v}_i^{\prime} \in \mathcal{S}_{\mathbf{v}_i}$, consider the corresponding submatrix $A_{\mathbf{v}_i^{\prime}}$ and obtain the set $\mathcal{A}_{\mathcal{S}_{\mathbf{v}_i}}$ of matrices corresponding to all possible $\mathbf{v}_i^{\prime} \in \mathcal{S}_{\mathbf{v}_i}$.
With slight abuse of notation, we refer to this set $\mathcal{A}_{\mathcal{S}_{\mathbf{v}_i}}$ as the cycle class of matrix $A_{\mathbf{v}_i}$.
Note that $A_{\mathbf{v}_i}$ can be chosen to be the representative of the set $\mathcal{A}_{\mathcal{S}_{\mathbf{v}_i}}$.
Let $\mathcal{A}_r \coloneqq \{A_{\mathbf{v}_1}, A_{\mathbf{v}_2}, \ldots, A_{\mathbf{v}_{N_r}} \}$ be the set of all these representative matrices.
Note that all possible ${n \choose r}$ submatrices of $B_{CRC}$ corresponding to $r$ non-stragglers can be partitioned into distinct cycle classes $\{ \mathcal{A}_{\mathcal{S}_{\mathbf{v}_1}}, \mathcal{A}_{\mathcal{S}_{\mathbf{v}_2}}, \ldots, \mathcal{A}_{\mathcal{S}_{\mathbf{v}_{N_r}}} \}$.
%
%
Using these notation, we next obtain the expected decoding error of the CRC scheme.

\input{Figures/Figure_Loss_curves_ISIT}

\begin{theorem}
\label{Theorem_CRC_EX_error_one_iteration}
For $r = 0, 1, \ldots, n$, consider the sets $\mathcal{V}_r, \mathcal{E}_r,$ and $\mathcal{A}_r$ of cardinality $N_r$ as defined above. 
%
%
%
Then for $(n,m)$-heterogeneous straggler model, the expected error $\mathbb{E}_{err}(CRC, n, m, s)$ of the $(s-1)$-tolerant, $n$-node CRC scheme with our random shuffling scheme is given by
\begin{align*}
\mathbb{E}_{err}(CRC, n, m, s) 
= \sum_{r = 0}^{n-s} P_r \sum_{i=1}^{N_r}\frac{e_i}{{n \choose r}} 
\lvert \lvert A_{\mathbf{v}_i} A_{\mathbf{v}_i}^{+} - \mathbf{1}_n \lvert\lvert_2^2
\end{align*}
where $P_r$ is defined in \cref{Theorem_FRC_EX_error_one_iteration} and  $A_{\mathbf{v}_i}^{+}$ denotes the pseudo inverse of $A_{\mathbf{v}_i}$.
%
%
\end{theorem}
\begin{proof}
As explained in \cref{Theorem_FRC_EX_error_one_iteration}, as a consequence of random shuffling, all possible ${n \choose r}$ submatrices $A$ of $B_{CRC}$ corresponding to $r$ non-stragglers are equally likely.
However, in general, the value of the decoding error $\mbox{err}(A)$ corresponding to each submatrix $A$ (obtained using \cref{Eqn_Error_submatrix_A}) could be different.
For the CRC scheme we next show that, the decoding error of any two submatrices that lie in the same cycle class is the same. 
The required expression for $\mathbb{E}_{err}(CRC, n, m, s)$ then follows from \cref{Eqn_Error_submatrix_A}, by conditioning over all possible ${n \choose r}$ submatrices and $r$.

Suppose two submatrices $A_1$ and $A_2$ belong to the same cycle class, say $\mathcal{A}_{\mathcal{S}_{\mathbf{v}_i}}$ for some $i \in \{1, 2, \ldots, N_r\}$.
For the CRC scheme it is known that, the $(j+1)$-th column of $B_{CRC}$ is obtained by one cyclic shift its $j$-th column for $j = 2,3,\ldots, n$ and the first column is obtained by one cyclic shift the $n$-th column~\cite{Exact_GD_First_2017}. 
Using this property and by the definition of cycle class,  
every column of $A_2$ is obtained by $m$ cyclic shifts of $A_1$ for some $m, 1 \leq m \leq n-1$. Thus the system of equations corresponding to $A_2 \mathbf{x} = \mathbf{1}_n$ is obtained by cycling shifting the system of equations $A_1 \mathbf{x} = \mathbf{1}_n$, since the right hand side of both the system of equations is equal to $\mathbf{1}_n$.
Since both the systems of equations are the same, they will have the same decoding error that can be obtained as $err(A_{\mathbf{v}_i}) = \lvert \lvert A_{\mathbf{v}_i} A_{\mathbf{v}_i}^{+} - \mathbf{1}_n \lvert\lvert_2^2$~\cite{Gilbert_Strang}.
\end{proof}

\begin{remark}
\label{Remark_total_expected_err}
Probability that $m$ workers out of $n$ belong to the slow-class is ${n \choose m} \hat{p}^m (1-\hat{p})^{n-m}) \eqqcolon P_m$ and
by conditioning over $m$, we can find expected error $\mathbb{E}_{err}(CRC, n, s)$ for heterogeneous straggler model as $\mathbb{E}_{err}(CRC, n, s) = \mathbb{E}_{err}(CRC, n, m, s) P_m$. Similarly for FRC scheme.
\hfill $\square$
\end{remark}

%
\section{Numerical results}
\label{Section_simulations}

In this section, we provide numerical results from a simulation environment  and from experiments on Amazon EC2 clusters for our proposed shuffling method for FRC and CRC schemes. 
To simulate the heterogeneous straggler model, 
we first classify workers into slow and active-classes 
and workers do not change their classes throughout the experiment of $L$ iterations. 
In any $\ell$-th iteration, slow and active class workers straggle with their assigned probabilities and we obtain a particular realization for the set of non-stragglers.
Using the computations performed by the given set of non-stragglers, we
train the data and plot the absolute loss versus iterations.


For the simulations, we have chosen data-sets, \textit{Boston} (linear regression), \textit{Diabetes} (logistic regression - binary classification), and \textit{Iris} (logistic regression, multi-class classification), provided in python's Scikit-Learn library.
For \textit{Boston} and \textit{Diabetes} data-sets, the plot of absolute loss versus iterations is provided in \cref{Figure_Loss_curves_ISIT}-(a) and (b) respectively, where we have chosen $n=8, s=2, \hat{p} = 0.3, p_{ss} = 0.8$ and $p_{as} = 0.01$.
It can be seen that the loss for CRC-shuffling is less than FRC-shuffling. This is because, for the chosen parameters we have $\mathbb{E}_{err}(CRC, n, s) = 0.5173$ and $\mathbb{E}_{err}(FRC, n, s) = 1.976$ 
(Theorems~\ref{Theorem_FRC_EX_error_one_iteration}, \ref{Theorem_CRC_EX_error_one_iteration}, and \cref{Remark_total_expected_err}).

For AWS experiments, we have chosen $t2.micro$ instances for 
workers and $m5.large$ for the master. 
Note that our heterogeneous straggler model is equivalent to master being waiting for a fixed-time duration of $t$ seconds for workers to finish their work. The workers who are unable to finish their computations before $t$ seconds are declared as stragglers. For our experiment, we choose this time $t$ in accordance with the chosen data-set. 
In \cref{Figure_Loss_curves_ISIT}-(c) and (d), the plot of absolute loss for training \textit{iris} data-set is illustrated for FRC and CRC respectively, where $n = 20, s= 2$. 
%

\section{Conclusion}
\label{Section_Conclusion}


In this paper, we propose the heterogeneous straggler model for DML to accurately model the correlated straggler behavior across iteration. The motivation for this model stems from our observations made on an Amazon AWS cluster.
%
%
Under this straggler model, we observe that the performance of existing DML algorithms is impaired, primarily because of non-uniform access of the training data. 
To improve the performance of DML algorithms, we propose data shuffling as a viable remedy. 
Shuffling strategy provides a simple, yet effective solution and is suitable for shared memory architectures (used in AWS) without causing any extra overhead.
Our experiments on the Amazon AWS cluster indeed attest the benefits of shuffling data when using approximate gradient coding schemes like FRC and CRC. 
We support these observations with theoretical expressions for the expected decoding error for FRC and CRC with shuffling, which could be of independent interest.

While the shuffling idea is elegant and easy to implement, as part of future work, it would be interesting to investigate existence of optimal data distribution strategies similar to those proposed in \cite{Timely_coded_Avestimehr_ISIT_2019, Timely_coded_Avestimehr_2019, Heterigenity_aware_2019} and compare its performance with our  shuffling scheme.
As part of future work, we also intend to analyze the performance of other approximate gradient coding schemes such as stochatic gradient coding~\cite{Approx_GD_SGD_Bitar_2020} and bernoulli gradient coding~\cite{Apprx_GD_sparse_graph_Papailiopoulos_2017_arxiv} with data shuffling. 
%


%
\section*{Acknowledgments}
This work is supported by the DST-INSPIRE faculty program of Government of India.


\bibliographystyle{IEEEtran}
\bibliography{References_Gradient_Coding}


\end{document}